\documentclass[conference]{IEEEtran}
\IEEEoverridecommandlockouts
\usepackage{cite}
\usepackage{amsthm,amsmath,amssymb,amsfonts,mathtools}
\usepackage{algorithm,algorithmic}
\usepackage{graphicx}
\usepackage{textcomp}
\usepackage{xcolor}
\def\BibTeX{{\rm B\kern-.05em{\sc i\kern-.025em b}\kern-.08em
    T\kern-.1667em\lower.7ex\hbox{E}\kern-.125emX}}

\newtheorem{thm}{Theorem}[section]
\newtheorem{cor}[thm]{Corollary}
\newtheorem{lem}[thm]{Lemma}
\newtheorem{expl}[thm]{Example}

\newtheorem{dfn}[thm]{Definition}
\newtheorem{rem}[thm]{Remark}

\newtheorem{hyp}[thm]{Hypothesis}

\newcommand{\del}{\mathrm{del}}
\newcommand{\ins}{\mathrm{ins}}
\newcommand{\rev}{\mathrm{rev}}
\newcommand{\BAD}{\mathrm{del_{BA}}}
\newcommand{\BAR}{\mathrm{rev_{BA}}}
\newcommand{\mono}{M_{a,m,{\bf k}}(n)}
\newcommand{\azinv}{A_{a,m}(n)}
\newcommand{\binset}{\mathbb{B}}
\newcommand{\Lzero}{L^{(0)}_{\bf k}(i,{\bf y})}
\newcommand{\Lone}{L^{(1)}_{\bf k}(i,{\bf y})}
\newcommand{\Rzero}{R^{(0)}_{\bf k}(i,{\bf y})}
\newcommand{\Rone}{R^{(1)}_{\bf k}(i,{\bf y})}
\newcommand{\monoweight}{\mathrm{wt}_{\bf k}({\bf y})}

\begin{document}

\title{Decoding algorithms of monotone codes and azinv codes and their unified view
}

\author{\IEEEauthorblockN{Hokuto Takahashi}
\IEEEauthorblockA{Graduate School of Science and Engineering, \\
Chiba University \\
1-33 Yayoi-cho, Inage-ku, 
Chiba-city, Chiba, Japan\\ 
263-0022\\
Email: ayca5495@chiba-u.jp}
\and
\IEEEauthorblockN{Manabu Hagiwara}
\IEEEauthorblockA{Graduate School of Science,\\
Chiba University\\
1-33 Yayoi-cho, Inage-ku, Chiba-city,
 Chiba, Japan\\ 
 263-0022\\
Email: hagiwara@math.s.chiba-u.ac.jp}

}

\maketitle

\begin{abstract}
This paper investigates linear-time decoding algorithms for two classes of error-correcting codes. One of the classes is monotone codes which are known as single deletion codes. The other is azinv codes which are known as single balanced adjacent deletion codes. As results, this paper proposes generalizations of Levenshtein's decoding algorithm for Levenshtein's single deletion codes. This paper points out that it is possible to unify our new two decoding algorithms.
\end{abstract}

\section{Introduction}
Insertion errors and deletion errors are considered to be
synchronization errors over communication channels and storage channels
such as DNA-data based storages \cite{buschmann2013levenshtein,gabrys2017codes}, 
racetrack memories\cite{chee2018coding,sima2019correcting}, 
or Bit Patterned Media\cite{inoue2011deletion,krishnan2011coding}.
The study of  deletion error-correcting codes started with Levenshtein's work\cite{levenshtein1966binary},
where he  proved that Varshamov-Tenengolts (VT) codes 
are capable of correcting single insertions or deletions.
In Levenshtein's proof, he showed beautiful decoding algorithm 
to correct single deletions.
This paper provides two generalizations of Levenshtein's
decoding algorithms. 
One corrects single deletions and reversals for monotone codes, 
the other corrects single balanced adjacent deletions and 
single balanced adjacent reversals for azinv codes.
A single Balanced Adjacent Deletion (BAD) error is a double deletion error that 
deletes two different consecutive binary symbols, i.e., 01 or 10. 
A single Balanced Adjacent Reversal (BAR) error is a swap error that
reverses two different consecutive binary symbols.

For applications, 
the computational cost of decoding algorithms is preferable to be polynomial.
One of remarkable aspects of our algorithms is the computational cost.
The costs are only linear to their input, i.e., the length of received word.
Another remarkable aspect is the following. 
While these classes and their correctable error are different, 
our algorithms can be allowed to be unified.

This paper is organized as follows. In Section \ref{monoalgo}, we first introduce 
monotone codes and provide a single deletion/reveral error-correcting algorithm of the monotone codes. 
In Section \ref{azinvalgo}, we first introduce 
azinv codes and provide a single BAD/BAR error-correcting algorithm of the azinv codes.
In Section \ref{comparison}, we provide a unified view of the algorithms 
of the monotone codes and the azinv codes.

\section{Monotone Code and its Decoding Algorithm}\label{monoalgo}
Throughout this paper, $\mathbb{B}$ denotes the binary set $\{0,1\}$.
For a positive integer $n$, $[n]$ denotes $\left\{1,2,\cdots,n\right\}$. 

In this section, we introduce monotone codes and provide an algorithm to make the monotone codes 
single deletion/reversal error-correctable (Algorithm 1). 
Errors treated in this section are deletion errors and reversal errors, which are defined below.

\begin{dfn}[Deletion and Reversal]
Let $n$ be a positive integer and $i \in [n]$.
Define a map $\del_i:\mathbb{B}^n \rightarrow \mathbb{B}^{n-1}$ as
\begin{equation*}
\del_i(x_1x_2\cdots x_n)  \coloneqq
x_1\cdots x_{i-1}x_{i+1}\cdots x_n  .
\end{equation*}
We call the map $\del_i$  {\bf deletion}.
Note that $\binset^0 := \{ \varepsilon \}$, where $\varepsilon$ is the empty word. 

Define a map $\rev_i:\mathbb{B}^n \rightarrow \mathbb{B}^n $ as
\begin{equation*}
\rev_i(x_1x_2\cdots x_n)  \coloneqq
x_1\cdots x_{i-1}\overline{x}_{i}x_{i+1}\cdots x_n 
\end{equation*}
with $\overline{0}=1$ and $\overline{1}=0$.
We call the map $\rev_i$  {\bf reversal}.
\end{dfn}

The following codes, monotone codes, are known 
as single deletion error-correcting codes \cite{hagiwara2016ordered}.
However, no decoding algorithm has been studied.

\begin{dfn}[Monotone code\cite{hagiwara2016ordered}]
Let $n$ and $m$ be positive integers, $a$ an integer, and ${\bf k}=(k_1,k_2,\cdots, k_n)$ 
a positive monotonic increasing integer sequence
with $m>k_n$.
Define a set $M_{a,m,{\bf k}}(n)$ as
\begin{equation*}
M_{a,m,{\bf k}}(n)\coloneqq \left\{ {\bf x}\in \binset^n \mid
     \rho_{{\bf k}}({\bf x}) \equiv a \pmod m \right\}, 
\end{equation*}
where
\begin{equation*} 
\rho_{{\bf k}}({\bf x}) \coloneqq \sum_{i=1}^{|{\bf x}| } k_ix_i 
\end{equation*}
and $|{\bf x}|$ denotes the length of  ${\bf x}$.
$\rho_{{\bf k}}({\bf x})$ is defined only for {\bf k} with $|{\bf k}|  \geq |{\bf x}|$.
We call  $M_{a,m,{\bf k}}(n)$ a \textbf{monotone code}.
\end{dfn}

\begin{rem}
If ${\bf k}=(1,2,\cdots, n)$, the monotone code is called a Levenshtein code.
If ${\bf k}=(1,2,\cdots, n)$ and $m=n+1$, the monotone code is called a VT code.
Levenshtein \cite{levenshtein1966binary} proved that Levenshtein codes are 
single deletion error-correcting codes with $m > n$.
He also proved that Levenshtein codes are single deletion/reversal error-correcting codes
with $m \ge 2n$.
\end{rem}

The following is one of the main contributions of this paper.

\begin{thm}
Let $\mono$ be a monotone code and ${\bf x}$ be a codeword of $\mono$. 
Let $\mathrm{Dec}_M$ denote Algorithm 1.

1. $\mathrm{Dec}_M ({\bf x}) = {\bf x}$.

2. Assume $k_n < m$. For any single deletion $\del$, 
$\mathrm{Dec}_{M} \circ \del ({\bf x}) = {\bf x}$.

3. If $k_n < m$, any monotone code $\mono$ is a single deletion error-correcting code with Algorithm 1 as a decoding algorithm.

4. Assume $2k_n \leq m$. For any single reversal $\rev$,
$\mathrm{Dec}_{M} \circ \rev ({\bf x}) = {\bf x}$.

5. If $2k_n \leq m$, any monotone code $\mono$ is a single deletion/reversal error-correcting code with Algorithm 1 as a decoding algorithm.
Here $\circ$ denotes the map composition.
\end{thm}
\begin{proof}
1. It is trivial from the steps 3 and 4 of Algorithm 1.
2. A proof is provided in subsection \ref{proof_del}.
3. It is a corollary of 1. and 2.
4. A proof is provided in subsection \ref{proof_rev}.
5. It is a corollary of  3. and 4.
 \end{proof}

\subsection{Decoding  algorithm for single deletion/reversal errors}
In this subsection, we provide notation in Algorithm 1 and the details of Algorithm 1. 

For  an integer $n\geq 1$,  a positive integer $i \in [n]$, 
a positive monotonic increasing integer sequence ${\bf k}=(k_1,k_2,\cdots, k_n)$, 
and a binary sequence ${\bf y}=y_1y_2\cdots y_{n-1}\in \binset^{n-1}$,
define maps as
\begin{align*}
&\Lzero \coloneqq
\begin{cases}
\displaystyle \sum_{j=1}^{i-1} \overline{y}_j(k_{j+1}-k_j)\quad &(i \neq 1) ,\\
0 &(i =1).
\end{cases}\\
&\Lone \coloneqq 
\begin{cases}
\displaystyle \sum_{j=1}^{i-1} y_j(k_{j+1}-k_j)\quad &(i \neq 1) ,\\
0 &(i =1).
\end{cases}\\
&\Rzero \coloneqq 
\begin{cases}
\displaystyle \sum_{j=i}^{n-1} \overline{y}_j(k_{j+1}-k_j)\quad &(i \neq n) , \\
0 &(i=n).
\end{cases}\\
&\Rone \coloneqq  
\begin{cases}
\displaystyle \sum_{j=i}^{n-1} y_j(k_{j+1}-k_j)\quad &(i \neq n), \\
0 &(i=n).
\end{cases}
\end{align*}
We denote $R^{(1)}_{\bf k}(1,{\bf y})$ by $\mathrm{wt}_{\bf k}({\bf y})$.
Note that $\mathrm{wt}_{\bf k}({\bf y})$ is coincided with the Hamming weight of {\bf y} if ${\bf k}=(1,2,\cdots, n).$
We omit ${\bf k}$ from the notations, if ${\bf k}=(1,2,\cdots, n)$. 

For ${\bf b} \in \binset$, $i \in [n]$, and $\ins_{i,{\bf b}}({\bf y})$, 
$L^{(0)}(i,{\bf y})$ (resp. $L^{(1)}(i,{\bf y})$) equals to the number of
$0$ (resp. $1$) to the left of inserted position $i$, $R^{(0)}(i,{\bf y})$ (resp. $R^{(1)}(i,{\bf y})$)
 equals to the number of $0$ (resp. $1$) 
to the right of inserted position $i$.

The map defined below is known as  an inverse operation of deletion 
and is used in Algorithm 1.

\begin{dfn}
For a positive integer $i \in [n+1]$ and a non-empty binary sequence ${\bf b}$,
define a map $\ins_{i,{\bf b}} : \binset^n \rightarrow \binset^{n+|{\bf b}|}$ as
\begin{equation*}
\ins_{i,{\bf b}}(x_1x_2\cdots x_n)\coloneqq\\
x_1\cdots x_{i-1}bx_{i}\cdots x_n .
\end{equation*}
We call the map $\ins_{i,{\bf b}}$  {\bf insertion}.
\end{dfn}
\begin{algorithm}
\caption{Decoding algorithm for single deletion/reversal errors}
\begin{algorithmic}[1]
\STATE Input: $a \in \mathbb{Z}_{\geq 0}$, $n$ and $m \in \mathbb{Z}_{\geq 1}$, 
${\bf y} \in \bigcup_{t \ge 0} \binset^t$, and 
${\bf k}=(k_1,k_2,\cdots,k_n)\in \mathbb{Z}^{n}$. \\
Output: ${\bf z} \in \binset^n$ or a symbol ?.
\STATE Compute the length of ${\bf y}$, say $|{\bf y}|$.
\IF {$|{\bf y}|  = n$} 
  \STATE go to 10.
\ELSIF{$|{\bf y}|  = n-1$}
  \STATE go to 21.
\ELSE  
  \STATE  ${\bf z} \coloneqq ?$. Go to 35.
\ENDIF
\STATE Compute $r \coloneqq \min \{s \in \mathbb{Z}_{\geq 0}
\mid s\equiv a-\rho_{\bf k}({\bf y}) \pmod {m} \}$.
\IF{$r=0$}
  \STATE ${\bf z} \coloneqq {\bf y}$. Go to 35.
\ELSE
  \STATE Compute $p \coloneqq \min \{j \in [n] \mid k_j = \min \{r,m-r\} \}$.
  \IF{$\rev_{p}({\bf y})\in \mono$}
    \STATE ${\bf z} \coloneqq \rev_{p}({\bf y})$. Go to 35.
  \ELSE
    \STATE ${\bf z} \coloneqq ?$. Go to 35.
  \ENDIF  
\ENDIF
\STATE Compute $r \coloneqq \min \{s \in \mathbb{Z}_{\geq 0}
 \mid s\equiv a-\rho_{\bf k}({\bf y}) \pmod {m} \}$.
\STATE Compute $w \coloneqq \mathrm{wt}_{\bf k}({\bf y})$.
\IF{$r\leq w$}
  \STATE Compute $p \coloneqq \max \{j \in [n] \mid R^{(1)}_{\bf k}(j,{\bf y})=r \}$.
  \STATE $b \coloneqq 0$.
\ELSE 
  \STATE Compute $p \coloneqq \min \{j \in [n] \mid
    L^{(0)}_{\bf k}(j,{\bf y})=r-w-k_1 \}.$
  \STATE $b \coloneqq 1$.
\ENDIF
\IF {$\ins_{p,{\bf b}}({\bf y}) \in \mono$}
  \STATE ${\bf z} \coloneqq \ins_{p,{\bf b}}({\bf y})$.
\ELSE
  \STATE ${\bf z} \coloneqq ?$.
\ENDIF
\STATE Output ${\bf z}$.
\end{algorithmic}
\end{algorithm}

Note that the steps 21st - 34th of Algorithm 1 are coincided with 
Levenshtein's decoding algorithm for single deletion errors 
if ${\bf k}=(1,2,\cdots, n)$  \cite{levenshtein1966binary}.

Before we move to proofs for Theorem II.4, 
we provide examples of Algorithm 1 for each error.

\begin{expl}[Decoding for a single deletion error]
Let $n=4, a=0, m=9, {\bf k}=(1,3,6,8),$ and ${\bf y}=101$.
Note that the monotone code has four words,
\begin{equation*} 
M_{0,9,{\bf k}}(4)=\{0000,1001,0110,1111\}.
\end{equation*}
\begin{itemize}
\item Since $|{\bf y}| =3$, we go to step 6 and then to step 21.
\item Algorithm 1 computes $r$ and $w$. We obtain
 $r=2$ and $w=4.$
\item Since $r\leq w$,
 \begin{align*} 
p&=\max \{j \in [4] \mid R^{(1)}(j,{\bf y})=2) \}\\
  &=3,\\
b&=0.
\end{align*}
\item The output $\mathrm{Dec}_{M}({\bf y})=\ins_{p,b}({\bf y})$ is  
\begin{equation*}
\ins_{3,0}(101)=1001.
\end{equation*}
\end{itemize}
\end{expl}

\begin{expl}[Decoding for a single reversal error]
Let $n=6, a=0, m=20, {\bf k}=(1,2,3,8,9,10),$ and ${\bf y}=111110$.
Note that the monotone code has five words,
\begin{equation*} 
M_{0,20,{\bf k}}(6)=\{000000,110110,001110,100011,010101\}.
\end{equation*}
\begin{itemize}
\item Since $|{\bf y}| =6$, we go to step 4 and then to step 10.
\item Algorithm 1 computes $r$, then we obtain $r=17$.
\item Since $r \neq 0$,
\begin{align*}
p&=\min\{3,17\} \\
  &=3.
\end{align*}
\item The output $\mathrm{Dec}_{M}({\bf y})=\rev_{p}({\bf y})$ is  
\begin{equation*}
\rev_{3}(111110)=110110.
\end{equation*}
\end{itemize}
\end{expl}

\subsection{Proof for single deletion error-correction}\label{proof_del}

To prove 2 of Theorem II.4, we introduce the following four Lemmas 
\ref{telescope}, \ref{weight}, \ref{inequality}, and \ref{remainder}.
From now on, till the end of this subsection, we assume the following.

\begin{hyp}\label{hyp_del} 
A binary sequence ${\bf x}$ is a codeword of $\mono$. 
Set ${\bf y} \coloneqq \del_i ({\bf x})$ for a fixed $i$.
$r$ is the value at the step 21 of Algorithm 1.
\end{hyp}

\begin{lem}\label{telescope}
$$k_i=k_1+\Lzero+\Lone.$$
\end{lem}
\begin{proof}
It follows from the definitions of $\Lzero$ and $\Lone$.
\begin{align*}
(\text{R.H.S.})
          &=k_1+\sum_{j=1}^{i-1} \overline{y}_j(k_{j+1}-k_j)+\sum_{j=1}^{i-1} y_j(k_{j+1}-k_j)\\
          &=k_1+\sum_{j=1}^{i-1} (\overline{y}_j+y_j)(k_{j+1}-k_j)\\
          &=k_1+\sum_{j=1}^{i-1} (k_{j+1}-k_j)\\
          &=k_i.
\end{align*}
\end{proof}

\begin{lem}\label{weight}
$\monoweight=\Lone+\Rone$.
\end{lem}
\begin{proof}
It follows from the definitions of $\monoweight , \Lone , $ and $\Rone$. 
\begin{align*}
(\text{R.H.S.})
          &=\sum_{j=1}^{i-1} y_j(k_{j+1}-k_j)+\sum_{j=i}^{n-1} y_j(k_{j+1}-k_j)\\
          &=\sum_{j=1}^{n-1} y_j(k_{j+1}-k_j)\\
          &=\monoweight.
\end{align*}
\end{proof}

\begin{lem}\label{inequality}
The following four inequalities hold.
\begin{align}
0&\leq \Rone \\ 
  &\leq \monoweight \\
  & < k_1+\monoweight+\Lzero \\ 
  &\leq k_n.
  \end{align}
\end{lem}
\begin{proof}
The inequality (1)  follows from that ${\bf k}$ is a positive monotonic increasing integer sequence.
The inequality (2) follows from Lemma \ref{weight} and $\Lone \geq 0$.
The inequality (3) follows from $k_1>0$ and $\Lzero \geq 0$.

We show the inequality (4).
From Lemma \ref{telescope} and Lemma \ref{weight}, the equation on the first line below follows.
\begin{align*}
k_1+\monoweight+\Lzero&=k_i+\Rone \\
          &=k_i+\sum_{j=i}^{n-1} y_j(k_{j+1}-k_j) \\
          &\leq k_i+\sum_{j=i}^{n-1} (k_{j+1}-k_j) \\
          &=k_n.
\end{align*}
\end{proof}

\begin{lem}\label{remainder}
\begin{equation*}r=
\begin{cases}
\Rone &(x_i=0),\\
k_1+\monoweight+\Lzero& (x_i=1).
\end{cases}
\end{equation*}
\end{lem}
\begin{proof}
The value of $r$ is obtained by using Lemma \ref{telescope}, 
Lemma \ref{weight} and Lemma \ref{inequality}.
It follows from the definitions of $r$ and $a$.
\begin{align*}
r&\equiv a-\rho_{{\bf k}}({\bf y}) \pmod m \\
 &= a-\sum_{i=1}^{n-1} k_iy_i  \\
 &\equiv \sum_{i=1}^n k_ix_i - \sum_{i=1}^{n-1} k_iy_i  \pmod m.
\end{align*}
By the assumption ${\bf y}=\mathrm{del}_i({\bf x})$, the following holds.
\begin{equation*}
{\bf x}=\ins_{i,x_i}({\bf y})=y_1\cdots y_{i-1}x_{i}y_{i}\cdots y_{n-1}.
\end{equation*}
Therefore,
\begin{align*}
&\sum_{i=1}^n k_ix_i - \sum_{i=1}^{n-1} k_iy_i\\
 &=k_1y_1+\cdots +k_{i-1}y_{i-1}+k_{i}x_{i}+k_{i+1}y_{i}\cdots +k_{n}y_{n-1}\\
 &-k_1y_1+\cdots +k_{i-1}y_{i-1}~~~~~~~~~+k_{i}y_{i}\cdots +k_{n-1}y_{n-1}\\
 &=k_ix_i+\sum_{j=i}^{n-1} y_j(k_{j+1}-k_j)\\
 &=k_ix_i+\Rone \\
 &=
 \begin{cases}
 \Rone &(x_i=0),\\
 k_i+\Rone&(x_i=1)\\
 \end{cases}\\
 &=
 \begin{cases}
 \Rone &(x_i=0),\\
 k_1+\monoweight+\Lzero&(x_i=1).
 \end{cases}
\end{align*}
The last follows from Lemma \ref{telescope} and Lemma \ref{weight}.
We show that the equality holds in the cases of  $x_i = 0$ and $x_i = 1$.

Case $x_i = 0$: we have shown 
\begin{equation*}
r \equiv \Rone \pmod m .
\end {equation*}
On the other hand, the following inequalities hold by Lemma \ref{inequality}.
\begin{equation*}
0 \leq \Rone <k_n <m. 
\end{equation*} 
By the definition of $r$,
\begin{equation*}
0\leq r<m.
\end{equation*}
This implies $r = \Rone$.

Case $x_i =1$: We have shown 
\begin{equation*}
r \equiv k_1 + \monoweight + \Lzero \pmod m .
\end{equation*}
On the other hand, the following inequalities hold by Lemma \ref{inequality}.
\begin{equation*}
0 \leq k_1+\monoweight+\Lzero \leq k_n<m.
\end{equation*}
By the definition of $r$,
\begin{equation*}
0 \leq r<m.
\end{equation*}

This implies $r=k_1+\monoweight+\Lzero$.

Therefore, Lemma \ref{remainder} holds. 
\end{proof}

\begin{proof}[proof of 2 of Theorem II.4]
Let us focus on the step 23 of Algorithm 1.
The case of $r \leq w$ and the case of $r > w$ are shown separately.

In the case of $r \leq w$, we have $\mathrm{Dec}_{M}({\bf y})=\ins_ {p, 0} ({\bf y})$.
We show $\ins_{p,0} ({\bf y}) = {\bf x}$.
First, we show that $r = \Rone$ holds. This is shown by contradiction.
From Lemma \ref {remainder}, either 
\begin{equation*}
r = \Rone \quad \text{or}
\end{equation*}
\begin{equation*} 
r = k_1 + \monoweight + \Lzero
\end{equation*}
holds.
Assume that $r = k_1 + \monoweight + \Lzero$ holds.
Lemma \ref{inequality} implies $r > w=\monoweight$, 
which contradicts to $r \leq w$.
Thus, $ r = \Rone $ holds.

Next, we show that the deleted symbol $ x_i $ is equal to $ 0 $.
This is also shown by contradiction.
Since $x_i \in \binset$, either $x_i =0 $ or $x_i =1$ holds.
If $x_i = 1$ holds, Lemma \ref{remainder} implies 
\begin{equation*}
r=k_1 + \monoweight + \Lzero, 
 \end{equation*} which contradicts to
\begin{equation*}
r=\Rone.
\end{equation*}
Therefore, $x_i = 0$.

Finally, we show that $\ins_ {p,0}({\bf y}) = {\bf x}$.
We showed the deleted symbol $x_i$ is equal to $0$. 
Since $x_i=0$ and ${\bf y} = \del_i ({\bf x}) $, ${\bf x} =\ins_ {i,0} ({\bf y})$ holds.
Therefore, it suffices to prove that $\ins_ {p,0} ({\bf y}) = \ins_ {i, 0} ({\bf y}) $. 
Furthermore, we will show 
\begin{equation*}
0=y_i=y_{i+1}=\cdots=y_{p-1}.
\end{equation*}
Since $r=\Rone$, then
\begin{equation*}
i\in \{j\in[n] \mid  R^{(1)}_{\bf k}(j,{\bf y})= r\}
\end{equation*}
holds. Since
\begin{equation*}
p\in \{j\in[n] \mid R^{(1)}_{\bf k}(j,{\bf y}) =r\},
\end{equation*}
then
\begin{align*}
r=\Rone=R^{(1)}_{\bf k}(p,{\bf y})
\end{align*}
holds. Therefore,
\begin{align*}
\sum_{j=i}^{n-1} y_j(k_{j+1}-k_j)=\sum_{j=p}^{n-1} y_j(k_{j+1}-k_j)
\end{align*}
holds. In the case of $r \leq w$, $i \leq p$ follows from the definition of $p$.
Since $i \leq p$, we have
\begin{align*}
0&=\sum_{j=i}^{n-1} y_j(k_{j+1}-k_j)-\sum_{j=p}^{n-1} y_j(k_{j+1}-k_j)\\
&=\sum_{j=i}^{p-1} y_j(k_{j+1}-k_j).
\end{align*}
Since $(k_ {j + 1} -k_j)> 0$, we have
\begin{equation*} 
0 = y_i = y_ {i + 1} = \cdots = y_{p-1}.
\end {equation*}

By a similar argument, we can prove in the remaining case 
$r > w$.
\end{proof}

\subsection{Proof for single reversal error-correction}\label{proof_rev}

To prove 4 of Theorem II.4, we introduce the following Lemma \ref{remainder_rev}.
From now on, till the end of this subsection, we assume the following.
\begin{hyp}\label{hyp_rev}
A binary sequence ${\bf x}$ is a codeword of $\mono$ with $m \geq 2k_n$. 
Set ${\bf y} \coloneqq \rev_i ({\bf x})$ for a fixed $i$.
$r$ is the value at the step 10 of Algorithm 1.
\end{hyp}

\begin{lem}\label{remainder_rev}
\begin{equation*}r=
\begin{cases}
k_i &(y_i=0),\\
m-k_i & (y_i=1)
\end{cases}
\end{equation*}
and $r \neq 0$.
\end{lem}   
\begin{proof}
It follows from the definitions of $r$ and $a$.
\begin{align*}
r&\equiv a-\rho_{{\bf k}}({\bf y}) \pmod m \\
&= a-\sum_{i=1}^{n} k_iy_i  \\
 &\equiv \sum_{i=1}^n k_ix_i - \sum_{i=1}^{n} k_iy_i  \pmod m.
\end{align*}
By the assumption ${\bf y}=\rev_i({\bf x})$, the following holds.
\begin{equation*}
{\bf x}=\rev_i({\bf y})=y_1\cdots y_{i-1}\overline{y}_iy_{i+1}\cdots y_n.
\end{equation*}
Therefore, we have
\begin{align*}
&\sum_{i=1}^n k_ix_i - \sum_{i=1}^{n} k_iy_i  \\
&=k_1y_1+\cdots +k_{i-1}y_{i-1}+k_{i}\overline{y}_{i}+k_{i+1}y_{i+1}\cdots +k_{n}y_{n}\\
&-k_1y_1+\cdots +k_{i-1}y_{i-1}+k_{i}y_{i}+k_{i+1}y_{i+1} \cdots +k_{n}y_{n}\\
&=k_i(\overline{y}_i - y_i)\\
&=
 \begin{cases}
 k_i &(y_i=0),\\
 -k_i&(y_i=1) \\
 \end{cases}\\
&\equiv
 \begin{cases}
 k_i &(y_i=0), \\
 m-k_i&(y_i=1). \\
 \end{cases}\\
\end{align*}
We show that the equality holds in the cases of $y_i=0$ and $y_i=1$. 
Since the sequence {\bf k} is a positive monotonic increasing integer sequence with $m \geq 2k_n$, 
we have $0<k_i\leq k_n <m$ and $0<m-k_n\leq m-k_i<m $. Therefore, 
\begin{equation*} 
r=
 \begin{cases}
 k_i &(y_i=0), \\
 m-k_i&(y_i=1) \\
 \end{cases}\\
 \end{equation*}
 and $r \neq 0$.
\end{proof}

\begin{proof}[proof of of 4 of Theorem II.4]
Lemma \ref{remainder_rev} implies $r \neq 0$. Therefore, $\mathrm{Dec}_{M}({\bf y})  \neq {\bf y}$.
Let us focus on the step 14 of Algorithm 1.
Lemma \ref{remainder_rev} implies $r=k_i$ or $r=m-k_i$. Whichever $r=k_i$ or $r=m-k_i$, 
\begin{align*}
\min \{r,m-r\}&=\min \{k_i, m-k_i\} \\
&=k_i
\end{align*}
holds, since m $\geq 2k_i$. Furthermore, for distinct indices $j_1$ and $j_2 \in [n]$, 
$k_{j_1} \neq k_{j_2}$ holds, since ${\bf k}$ is a positive monotonic increasing integer sequence.   
Therefore, we have 
\begin{align*}
p&=\min \{j \in [n] \mid k_j = \min \{r,m-r\} \} \\
  &=i 
\end{align*}
holds. Thus, 
\begin{align*}
\rev_{p}({\bf y}) 
&= \rev_{i}({\bf y}) \\
&={\bf x} \\
& \in \mono .
\end{align*} 
Therefore,  $\mathrm{Dec}_{M}({\bf y})=\rev_{p}({\bf y}) ={\bf x}$.
\end{proof}

\section{Azinv Code and its Decoding Algorithm}\label{azinvalgo}

In this section, we provide an algorithm to make azinv codes 
 single BAD/BAR error-correctable (Algorithm 2).
Errors treated in this section 
are BAD errors and BAR errors, which are defined below.

\begin{dfn}[BAD and BAR]
For an integer $n \geq 2$  and  $i \in [n-1]$, 
define a partial map $\mathrm{BD}_i:\binset^n \rightarrow \binset^{n-2}$ as
\begin{equation*}
\mathrm{BD}_i(x_1x_2\cdots x_n)\coloneqq x_1\cdots x_{i-1}x_{i+2}\cdots x_{n} 
\end{equation*}
only for ${\bf x}$ with  $x_i\neq x_{i+1}$.
We call the partial map $\mathrm{BD}_i$  {\bf balanced adjacent deletion (BAD)}.

Define a partial map $\mathrm{BR}_i:\binset^n \rightarrow \binset^n $ as
\begin{equation*}
\mathrm{BR}_i(x_1x_2\cdots x_n)  \coloneqq
x_1\cdots x_{i-1}\overline{x}_{i}\overline{x}_{i+1}\cdots x_n 
\end{equation*}
only for ${\bf x}$ with $x_i\neq x_{i+1}$.
We call the partial map $\mathrm{BR}_i$  {\bf balanced adjacent reversal (BAR)}.
\end{dfn}

The following codes, azinv codes,  are known as single BAD error-correcting codes \cite{hagiwara2017perfect}.
However, no decoding algorithm has been studied.

\begin{dfn}[Azinv code \cite{hagiwara2017perfect}]
For integers $n \geq 2$ and $m \geq 2$ and an integer $a$, 
define a set $\azinv$ as
\begin{align*}
\azinv\coloneqq \{ {\bf x} \in \binset^{n} \mid \tau({\bf x})
\equiv a \pmod m,  {\bf x} \neq {\bf 0},{\bf 1}\}
\end{align*}
with $m \ge n$, where ${\bf 0}$ (resp. ${\bf 1}$) is the all zero (resp. one) word,
the function $\tau$ is the composition of the function $\mathrm{inv}$ below 
and the permutation $\sigma^{-1}$ below i.e., 
$\tau \coloneqq \mathrm{inv}  \circ \sigma^{-1}$.

The function $\mathrm{inv}$ is a map from a binary word to a non-negative integer
and is defined as
\begin{equation*}
\mathrm{inv}(x_1x_2\cdots x_n)\coloneqq \# \{(i,j) \mid 1\leq i < j \leq n, x_i > x_j \}.
\end{equation*}
The value $\mathrm{inv}({\bf x})$ is called the inversion number of ${\bf x}$.

The permutation $\sigma$ is defined as

$\sigma(x_1x_2\cdots x_n)\coloneqq$
\begin{equation*}
\begin{cases}
x_1x_nx_2x_{n-1}x_3\cdots x_{\frac{n+4}{2}}x_{\frac{n}{2}}x_{\frac{n+2}{2}}     &(n: \text{even}),\\
x_1x_nx_2x_{n-1}x_3\cdots x_{\frac{n-1}{2}}x_{\frac{n+3}{2}}x_{\frac{n+1}{2}} &(\text{otherwise}).
\end{cases}
\end{equation*}
We call $\azinv$ an azinv code.
\end{dfn}

The following is one of the main contributions of this paper.
\begin{thm}
Let $\azinv$ be an azinv code and ${\bf x}$ be a codeword of $\azinv$. 
Let $\mathrm{Dec}_A$ denote Algorithm 2.

1. $\mathrm{Dec}_A({\bf x}) = {\bf x}$.

2. Assume $n \leq m$. For any single BAD $\BAD$, 
$\mathrm{Dec}_{A} \circ \BAD ({\bf x}) = {\bf x}$.

3. If $n \leq m$, any azinv code $\azinv$ is a single BAD error-correcting code with Algorithm 2 as a decoding algorithm.

4. Assume $2(n-1) \leq m$. For any single BAR $\BAR$,
$\mathrm{Dec}_{A} \circ \BAR ({\bf x}) = {\bf x}$.

5. If $2(n-1) \leq m$, any monotone code $\azinv$ is a single BAD/BAR error-correcting code with Algorithm 2 as a decoding algorithm.
Here $\circ$ denotes the map composition.
\end{thm}
\begin{proof}
1. It is trivial from the steps 3 and 4 of Algorithm 2.
2. A proof is provided in subsection \ref{proof_BAD}.
3. It is a corollary of 1. and 2.
4. A proof is provided in subsection \ref{proof_BAR}.
5. It is a corollary of 3. and 4.
 \end{proof}

\subsection{Decoding algorithm for single BAD errors}
In this subsection, we provide the details of Algorithm 2.
Curiously, Algorithm 2 is similar to Algorithm 1 for correcting single 
 deletion/reversal errors for monotone codes.

We introduce the following notation.
For a positive integer $n$ and a binary sequence ${\bf y}=y_1y_2\cdots y_n$,
define ${\bf \tilde{y}}$ as
\begin{equation*}
{\bf \tilde{y}} \coloneqq 
\begin{cases}
y_1\overline{y}_2y_3\cdots y_{n-1}\overline{y}_{n} &(n: \text{even}),\\
y_1\overline{y}_2y_3\cdots \overline{y}_{n-1}y_{n} &(\text{otherwise}).
\end{cases}
\end{equation*}
For the $i$th entry of ${\bf y}$, say $y_i$, define $\tilde{y}_i$ as \\
\begin{equation*}
\tilde{y}_i \coloneqq
\begin{cases}
\overline{y}_i & (i: \text{even}), \\
y_i &(\text{otherwise}).  
\end{cases}
\end{equation*}
For integers $i$ and $j$, define $[i,j]$ as 
$[i,j]\coloneqq \{k \in \mathbb{Z}|i \leq k \leq j \}$.
We denote the subsequence of ${\bf y}$ in the range $[i,j]$, by ${\bf y}_{[i,j]}$, i.e.,  
${\bf y}_{[i,j]} \coloneqq  y_i y_{i+1} \cdots y_j$.
By using the range notation, the permutation $\sigma^{-1}$ can be written in the following form.
\begin{rem}
$\sigma^{-1}(y_1y_2 \cdots y_n)=y_1 \sigma^{-1}({\bf y}_{[3,n]}) y_2$.
\end{rem}

\begin{algorithm}
\caption{Decoding algorithm for single BAD/BAR errors}
\begin{algorithmic}[1]
\STATE Input: $a \in \mathbb{Z}_{\geq 0}, n$ and $m \in \mathbb{Z}_{\geq 2}$, and 
${\bf y} \in \bigcup_{t \ge 0} \binset^t$. \\ 
Output: ${\bf z} \in \binset^n$ or a symbol ?. 
\STATE Compute the length of ${\bf y}$, say $|{\bf y}|$.
\IF {$|{\bf y}|  = n$} 
  \STATE Go to 10.
\ELSIF {$|{\bf y}| = n-2$}
  \STATE go to 21.
\ELSE  
  \STATE  ${\bf z} \coloneqq ?$. Go to 43.
\ENDIF
\STATE Compute $r \coloneqq \min \{s \in \mathbb{Z}_{\geq 0}
\mid s\equiv a-\tau({\bf y}) \pmod {m} \}$.
\IF{$r=0$}
  \STATE ${\bf z} \coloneqq {\bf y}$. Go to 43.
\ELSE
  \STATE Compute $p \coloneqq n - \min \{r,m-r\}$.
  \IF{$\BAR_{p}({\bf y}) \in \azinv$}
    \STATE ${\bf z} \coloneqq  \BAR_{p}({\bf y})$. Go to 43.
  \ELSE
    \STATE ${\bf z} \coloneqq ?$.  Go to 43.
  \ENDIF  
\ENDIF
\STATE Compute $r \coloneqq \min \{s \in \mathbb{Z}_{\geq 0}
 \mid s\equiv a-\tau({\bf y}) \pmod {m} \}$.
\STATE Compute $w \coloneqq \mathrm{wt}({\bf \tilde{y}})$.
\IF{$r\leq w$}
  \STATE Compute $p \coloneqq \max \{j \in [n-1] \mid L^{(1)}(j,{\bf \tilde{y}})=r) \}$.
  \IF{$p$: even}
    \STATE $b \coloneqq 10$.
  \ELSE
    \STATE $b  \coloneqq 01$.
  \ENDIF  
\ELSE 
  \STATE Compute \\ $p \coloneqq \min \{j \in [n-1] \mid 
  R^{(0)}(j,{\bf \tilde{y}})=r-w-1) \}$.
  \IF{$p$: even}  
    \STATE $b \coloneqq 01$.
  \ELSE
    \STATE $b \coloneqq 10$.
  \ENDIF    
\ENDIF
\IF {$\ins_{p,{\bf b}}({\bf y}) \in \azinv$}
  \STATE ${\bf z} \coloneqq \ins_{p,{\bf b}}({\bf y})$.
\ELSE
  \STATE ${\bf z} \coloneqq ?$. 
\ENDIF
\STATE Output {\bf z}.
\end{algorithmic}
\end{algorithm}

Before we move to proofs for Theorem III.3, we provide examples of Algorithm 2 for each error.

\begin{expl}[Decoding for a single BAD error]
Let $n=5, a=0, m=5,$ and ${\bf y}=101$.
Note that the azinv code has six words,
\begin{equation*} 
C_{0,5}(5)=\{01000,01010,01010,01011,01111,10110,10001\}.
\end{equation*}
\begin{itemize}
\item Since $|{\bf y}| =3$, we go to step 6 and then to step 21.
\item Algorithm 2 computes $r$ and $w$, 
then we obtain $r=3$ and $w=3.$
\item Since $r \leq w$, then
\begin{align*}
p&=\max \{j \in [4] \mid L^{(1)}(j,{\bf \tilde{y}})=3) \} \\
  &=4.
\end{align*}
\item Since $p=4$ is even, then $b=10$.
\item The output $\mathrm{Dec}_{A}({\bf y})=\ins_{p,b}({\bf y})$ is  
\begin{equation*}
\ins_{4,10}(101)=10110.
\end{equation*}
\end{itemize}
\end{expl}

\begin{expl}[Decoding for a single BAR error]
Let $n=6, a=0, m=10,$ and ${\bf y}=100000$.
Note that the azinv code has five words,
\begin{equation*} 
C_{0,10}(6)=\{010000,010100,010101,010111,011111\}.
\end{equation*}
\begin{itemize}
\item Since $|{\bf y}| =6$, we go to step 4 and then to step 10.
\item Algorithm 2 computes $r$, then $r=5$.
\item Since $r \neq 0$, 
\begin{align*}
p&=6-\min\{5,5\} \\
  &=1.
\end{align*}
\item The output $\mathrm{Dec}_{A}({\bf y})=\BAR_{p}({\bf y})$ is  
\begin{equation*}
\mathrm{BR}_{1}(100000)=010000.
\end{equation*}
\end{itemize}
\end{expl}

\subsection{Proof for single BAD error-correction}\label{proof_BAD}

To prove 2 of Theorem III.3, we introduce the following three Lemmas \ref{inequality_BAD}, 
\ref{partion_BAD} and \ref{remainder_BAD}
From now on, till the end of this subsection, we assume the following.
\begin{hyp}\label{hyp_BAD}
A binary sequence ${\bf x}$ is a codeword of $\azinv$. 
Set ${\bf y} \coloneqq \BAD_i ({\bf x})$ for a fixed $i$.
 $r$ is the value at the step 21 of Algorithm 2.
\end{hyp}

\begin{lem}\label{inequality_BAD}
The following four inequalities hold.
\begin{align}
\setcounter{equation}{0}
0&\leq L^{(1)}(i,{\bf y})\\ 
  &\leq \mathrm{wt}({\bf y})\\
  &< 1+\mathrm{wt}({\bf y})+R^{(0)}(i,{\bf y})\\ 
  &< n.
  \end{align}
\end{lem}
\begin{proof}
The inequality (1) follows from the definition of $L^{(1)}(i,{\bf y})$.
The inequality (2) follows from Lemma \ref{weight} and $R^{(1)}(i,{\bf y}) \geq 0$.
The inequality (3) follows from $R^{(0)}(i,{\bf y}) \geq 0$.
We show the inequality (4).
\begin{align*}
1+\mathrm{wt}({\bf y})+R^{(0)}(i,{\bf y})&=1+R^{(1)}(1,{\bf y})+R^{(0)}(i,{\bf y})\\
          &\leq 1+R^{(1)}(1,{\bf y})+R^{(0)}(1,{\bf y})\\
          &= n-1 \\
          &< n.
\end{align*}
\end{proof}

\begin{lem}\label{partion_BAD}
Either $\tilde{x}_i \tilde{x}_{i+1}=00$ or $\tilde{x}_i \tilde{x}_{i+1}=11$ holds.
\end{lem}
\begin{proof}
By the assumption ${\bf y} = \BAD_i ({\bf x})$, $x_i \neq x_{i+1}$ holds.
Thus, either $x_ix_{i+1}=01$ or $x_{i}x_{i+1}=10$ holds.
Therefofre, whichever $i$ is odd or even, 
either $\tilde{x}_i \tilde{x}_{i+1}=00$ or $\tilde{x}_i \tilde{x}_{i+1}=11$ holds.
\end{proof}

\begin{lem}\label{remainder_BAD}
\begin{equation*}r=
\begin{cases}
L^{(1)}(i,{\bf \tilde{y}}) &(\tilde{x}_i \tilde{x}_{i+1}=00),\\
1+\mathrm{wt}({\bf \tilde{y}})+R^{(0)}(i,{\bf \tilde{y}}) & (\tilde{x}_i \tilde{x}_{i+1}=11).
\end{cases}
\end{equation*}
\end{lem}

\begin{proof}
It follows from the definitions of $r$ and $a$.
\begin{align*}
r &\equiv a - \tau({\bf y}) \pmod m \\
  & = a - \mathrm{inv}(\sigma^{-1}({\bf y})) \\
  & \equiv \mathrm{inv}(\sigma^{-1}({\bf x})) - \mathrm{inv}(\sigma^{-1}({\bf y}))  \pmod m \\
 \end{align*}
 By the assumption ${\bf y}=\mathrm{BAD}_i({\bf x})$, the following holds.
 \begin{equation*}
{\bf x}=\ins_{i,x_ix_{i+1}}({\bf y})=y_1y_2\cdots y_{i-1}x_{i}x_{i+1}y_{i}y_{i+1}\cdots y_{n-2}.
\end{equation*}
Therefore, we have
\begin{align*}
&\mathrm{inv}(\sigma^{-1}({\bf x})) - \mathrm{inv}(\sigma^{-1}({\bf y})) \\
  & = \mathrm{inv}(\sigma^{-1}(\ins_{i,x_ix_{i+1}}({\bf y}))) - \mathrm{inv}(\sigma^{-1}({\bf y})) \\
  &= \mathrm{inv}(\sigma^{-1}(y_1y_2\cdots y_{i-1} x_ix_{i+1}y_iy_{i+1}\cdots y_{n-2})) \\
  &- \mathrm{inv}(\sigma^{-1}(y_1y_2\cdots y_{i-1}\quad \quad \quad y_iy_{i+1}\cdots y_{n-2})) \\ 
  \end{align*}
\begin{align*}
&= \left \{
\begin{array}{cc}
 \mathrm{inv}(y_1 y_3 \cdots y_{i-2} x_i  \sigma^{-1}({\bf y}_{[i,n-2]}) x_{i+1} y_{i-1}\cdots y_4 y_2) \\ \\
  - \mathrm{inv}(y_1 y_3 \cdots y_{i-2} \quad \sigma^{-1}({\bf y}_{[i,n-2]}) \quad y_{i-1}\cdots y_4 y_2)\\
  (i:\text{odd}), \\ \\ \\
\mathrm{inv}(y_1 y_3 \cdots y_{i-1} x_{i+1}  \sigma^{-1}({\bf y}_{[i+1,n-2]}) y_i x_i \cdots y_4 y_2) \\ \\
  - \mathrm{inv}(y_1 y_3 \cdots y_{i-1} \quad \sigma^{-1}({\bf y}_{[i+1,n-2]}) y_i \quad \cdots y_4 y_2)\\
  (i:\text{even})
\end{array}
\right. \\
\end{align*}
\begin{align*}
&= \left \{
\begin{array}{cc}
 \mathrm{inv}(y_1 y_3 \cdots y_{i-2} \quad 0  \sigma^{-1}({\bf y}_{[i,n-2]}) 1 \quad y_{i-1}\cdots y_4 y_2) \\ \\
  - \mathrm{inv}(y_1 y_3 \cdots y_{i-2} \quad \sigma^{-1}({\bf y}_{[i,n-2]}) \quad y_{i-1}\cdots y_4 y_2)\\
  (i:\text{odd}, x_ix_{i+1}=01), \\ \\ \\
  \mathrm{inv}(y_1 y_3 \cdots y_{i-1}\quad 0  \sigma^{-1}({\bf y}_{[i+1,n-2]}) y_i 1\quad \cdots y_4 y_2) \\ \\
  - \mathrm{inv}(y_1 y_3 \cdots y_{i-1} \quad \sigma^{-1}({\bf y}_{[i+1,n-2]}) y_i \quad \cdots y_4 y_2)\\
  (i:\text{even}, x_ix_{i+1}=10), \\ \\ \\
  \mathrm{inv}(y_1 y_3 \cdots y_{i-2}\quad 1  \sigma^{-1}({\bf y}_{[i,n-2]}) 0 \quad y_{i-1}\cdots y_4 y_2) \\ \\
  - \mathrm{inv}(y_1 y_3 \cdots y_{i-2} \quad \sigma^{-1}({\bf y}_{[i,n-2]}) \quad y_{i-1}\cdots y_4 y_2)\\
  (i:\text{odd}, x_ix_{i+1}=10), \\ \\ \\
\mathrm{inv}(y_1 y_3 \cdots y_{i-1}\quad 1  \sigma^{-1}({\bf y}_{[i+1,n-2]}) y_i 0 \quad \cdots y_4 y_2) \\ \\
  - \mathrm{inv}(y_1 y_3 \cdots y_{i-1} \quad \sigma^{-1}({\bf y}_{[i+1,n-2]}) y_i \quad \cdots y_4 y_2)\\
  (i:\text{even}, x_ix_{i+1}=01)
\end{array}
\right. \\
\end{align*}
\begin{align*}
&= \left \{
\begin{array}{ll}
\displaystyle \sum_{j=1,j:\text{odd}}^{i-2} y_j + \sum_{j=2,j:\text{even}}^{i-1} \overline{y}_j  \\
(i:\text{odd}, x_ix_{i+1}=01), \\ \\
\displaystyle \sum_{j=1,j:\text{odd}}^{i-1} y_j + \sum_{j=2,j:\text{even}}^{i-2} \overline{y}_j  \\
(i:\text{even}, x_ix_{i+1}=10), \\ \\
\displaystyle \sum_{j=1,j:\text{odd}}^{i-2} y_j + \sum_{j=2,j:\text{even}}^{i-1} \overline{y}_j  + (n-2-(i-1))+1\\
 (i:\text{odd}, x_ix_{i+1}=10), \\ \\
 \displaystyle \sum_{j=1,j:\text{odd}}^{i-1} y_j + \sum_{j=2,j:\text{even}}^{i-2} \overline{y}_j  + (n-2-(i-1))+1\\
 (i:\text{even}, x_ix_{i+1}=01) \\
\end{array}
\right. \\
\end{align*}
\begin{align*}
&= \left \{
\begin{array}{cc}
\displaystyle \sum_{j=1}^{i-1} \tilde{y}_j \\
(i:\text{odd}, x_ix_{i+1}=01 \quad \text{or} \quad i:\text{even}, x_ix_{i+1}=10), \\  \\
\displaystyle \sum_{j=1}^{i-1} \tilde{y}_j + \sum_{j=i}^{n-2} (\tilde{y}_j+\overline{\tilde{y}}_j) +1 \\
 (i:\text{odd}, x_ix_{i+1}=10 \quad \text{or} \quad i:\text{even}, x_ix_{i+1}=01)
\end{array}
\right. \\ \\
&= \left \{
\begin{array}{cc}
L^{(1)}(i,{\bf \tilde{y}}) \\
(\tilde{x}_i \tilde{x}_{i+1}=00), \\  \\
L^{(1)}(i,{\bf \tilde{y}})+R^{(1)}(i,{\bf \tilde{y}})+R^{(0)}(i,{\bf \tilde{y}})+1 \\
 (\tilde{x}_i \tilde{x}_{i+1}=11)
\end{array}
\right. \\ 
&= \left \{
\begin{array}{cc}
L^{(1)}(i,{\bf \tilde{y}}) &(\tilde{x}_i \tilde{x}_{i+1}=00), \\
1+\mathrm{wt}({\bf \tilde{y}})+R^{(0)}(i,{\bf \tilde{y}}) & (\tilde{x}_i \tilde{x}_{i+1}=11).
\end{array}
\right. \\
\end{align*} 

The last follows from Lemma \ref{weight}.
We show that the equality holds in the cases of  $\tilde{x}_i\tilde{x}_{i+1} = 00$ and $\tilde{x}_i\tilde{x}_{i+1} = 11$.

Case $\tilde{x}_i\tilde{x}_{i+1} = 00$: we have shown 
\begin{equation*}
r \equiv L^{(1)}(i,{\bf \tilde{y}}) \pmod m .
\end {equation*}
On the other hand, the following inequalities hold by Lemma \ref{inequality_BAD}.
\begin{equation*}
0 \leq  L^{(1)}(i,{\bf \tilde{y}}) < n <m. 
\end{equation*} 
By the definition of $r$, 
\begin{equation*}
0 \leq r <m.
\end{equation*}
This implies $r = L^{(1)}(i,{\bf \tilde{y}})$.

Case $\tilde{x}_i\tilde{x}_{i+1} = 11$: 
We have shown 
\begin{equation*}
r \equiv 1+\mathrm{wt}({\bf \tilde{y}})+R^{(0)}(i,{\bf \tilde{y}}) \pmod m .
\end{equation*}
On the other hand, the following inequalities hold by Lemma \ref{inequality_BAD}.
\begin{equation*}
0 \leq 1+\mathrm{wt}({\bf \tilde{y}})+R^{(0)}(i,{\bf \tilde{y}}) \leq n < m.
\end{equation*}
By the definition of $r$, 
\begin{equation*}
0 \leq r < m.
\end{equation*}
This implies $r=1+\mathrm{wt}({\bf \tilde{y}})+R^{(0)}(i,{\bf \tilde{y}})$.
Therefore, Lemma \ref{remainder_BAD} holds.
\end{proof}

\begin{proof}[proof of 2 of Theorem III.3]
Let us focus on the step 23 of Algorithm 2.
The case of $r \leq w$ and the case of $r > w$ are shown separately.

In the case of $r \leq w$, we have $\mathrm{Dec}_{A}({\bf y}) = \ins_ {p, b} ({\bf y})$.
We show $\ins_ {p, b} ({\bf y}) = {\bf x}$. 
First, we can show that $r = L^{(1)}(i,{\bf \tilde{y}})$ holds, 
similarly to the case of Monotone codes.

Next, we show that the deleted symbols $ x_ix_{i+1}$ satisfy $\tilde{x}_i\tilde{x}_{i+1}= 00$.
This is shown by contradiction.
Since $x_ix_{i+1} \in \{01,10 \}$, either $\tilde{x}_i\tilde{x}_{i+1} =00 $ or $\tilde{x}_i\tilde{x}_{i+1} =11$ holds.
Assume that $\tilde{x}_i\tilde{x}_{i+1} = 11$ holds. Lemma \ref{remainder_BAD} implies 
\begin{equation*}
r=1+\mathrm{wt}({\bf \tilde{y}})+R^{(0)}(i,{\bf \tilde{y}}),
 \end{equation*} which contradicts to
\begin{equation*}
r=L^{(1)}(i,{\bf \tilde{y}}).
\end{equation*}
Therefore, $\tilde{x}_i\tilde{x}_{i+1} = 00$.

Finally, we show that $\ins_ {p,b}({\bf y}) = {\bf x}$.
We showed that the deleted symbols $x_ix_{i+1}$ satisfy $\tilde{x}_i\tilde{x}_{i+1}=00$.
By the assumption ${\bf y} = \BAD_i ({\bf x})$,  we have ${\bf x} =\ins_ {i,x_ix_{i+1}} ({\bf y})$. 
Therefore, it suffices to prove that $\ins_ {p,b} ({\bf y}) = \ins_ {i, x_ix_{i+1}} ({\bf y})$. 
Furthermore we will show 
\begin{equation*}
0=\tilde{y}_i=\tilde{y}_{i+1}=\cdots=\tilde{y}_{p-1}.
\end{equation*}

Since $r=L^{(1)}(i,{\bf \tilde{y}})$, then
\begin{equation*}
i\in \{j \in [n-1] \mid L^{(1)}(j,{\bf \tilde{y}})=r) \}
\end{equation*}
holds. Since
\begin{equation*}
p\in \{j \in [n-1] \mid L^{(1)}(j,{\bf \tilde{y}})=r) \},
\end{equation*}
then
\begin{align*}
r=L^{(1)}(i,{\bf \tilde{y}})=L^{(1)}(p,{\bf \tilde{y}})
\end{align*}
holds. Therefore,
\begin{align*}
\sum_{j=1}^{i-1} \tilde{y}_j=\sum_{j=1}^{p-1} \tilde{y}_j
\end{align*}
holds. In the case of $r \leq w$,  $i \leq p$ follows from the definition of $p$. 
Since $i \leq p$, then
\begin{align*}
0&=\sum_{j=1}^{p-1} \tilde{y}_j-\sum_{j=1}^{i-1} \tilde{y}_j\\
&=\sum_{j=i}^{p-1} \tilde{y}_j
\end{align*}
holds. Therefore, we have
\begin{equation*}
0=\tilde{y}_i=\tilde{y}_{i+1}=\cdots=\tilde{y}_{p-1}.
\end{equation*}

By a similar argument, we can prove in the remaining case $ r > w$. 
\end{proof}

\subsection{Proof for single BAR error-correction}\label{proof_BAR}

To prove 4 of Theorem III.3, we introduce the following two Lemmas \ref{partion_BAR} and
\ref{remainder_BAR}.
From now on, till the end of this subsection, we assume the following.
\begin{hyp}\label{hyp_BAR}
A binary sequence ${\bf x}$ is a codeword of $\azinv$ with $m \geq 2(n-1)$. 
Set ${\bf y} \coloneqq \BAR_i ({\bf x})$ for a fixed $i$.
$r$ is the value at the step 10 of Algorithm 2.
\end{hyp}

\begin{lem}\label{partion_BAR}
Either $\tilde{y}_i \tilde{y}_{i+1}=00$ or $\tilde{y}_i \tilde{y}_{i+1}=11$ holds.
\end{lem}
\begin{proof}
By the assumption ${\bf y} = \BAR_i ({\bf x})$, 
$x_i \neq x_{i+1}$ holds.
Hence, $y_i \neq y_{i+1}$ holds.
Thus, either $y_iy_{i+1}=01$ or $y_{i}y_{i+1}=10$ holds.
Therefofre, whichever $i$ is odd or even, 
either $\tilde{y}_i \tilde{y}_{i+1}=00$ or $\tilde{y}_i \tilde{y}_{i+1}=11$ holds.
\end{proof}

\begin{lem}\label{remainder_BAR}
\begin{equation*}r=
\begin{cases}
n-i &(\tilde{y}_i \tilde{y}_{i+1}=00),\\
m-(n-i)& (\tilde{y}_i \tilde{y}_{i+1}=11)
\end{cases}
\end{equation*}
and $r \neq 0$.
\end{lem}   
\begin{proof}
It follows from the definitions of $r$ and $a$.
\begin{align*}
r &\equiv a - \tau({\bf y}) \pmod m \\
  & = a - \mathrm{inv}(\sigma^{-1}({\bf y})) \\
  & \equiv \mathrm{inv}(\sigma^{-1}({\bf x})) - \mathrm{inv}(\sigma^{-1}({\bf y})) \pmod m \\
\end{align*}
By the assumption ${\bf y}=\BAR_i({\bf x})$, the following holds.
\begin{equation*}
{\bf x}=\BAR_i({\bf y})=y_1y_2\cdots y_{i-1}{y}_{i+1}y_i\cdots y_n.
\end{equation*}
Therefore, we have
\begin{align*}
&\mathrm{inv}(\sigma^{-1}({\bf x})) - \mathrm{inv}(\sigma^{-1}({\bf y})) \\
  & = \mathrm{inv}(\sigma^{-1}(\BAR_i({\bf y}))) - \mathrm{inv}(\sigma^{-1}({\bf y})) \\
  &= \mathrm{inv}(\sigma^{-1}(y_1y_2\cdots y_{i-1}y_{i+1}y_i \cdots y_{n})) \\
  &- \mathrm{inv}(\sigma^{-1}(y_1y_2\cdots y_{i-1}y_iy_{i+1}\cdots y_{n})) \\
  &= \left \{
\begin{array}{cc}
 \mathrm{inv}(y_1 y_3 \cdots y_{i+1} \sigma^{-1}({\bf y}_{[i+2,n]}) y_i y_{i-1}\cdots y_4 y_2) \\ \\
  - \mathrm{inv}(y_1 y_3 \cdots y_i \quad \sigma^{-1}({\bf y}_{[i+2,n]}) y_{i+1} y_{i-1}\cdots y_4 y_2) \\
  (i:\text{odd}), \\ \\ \\
\mathrm{inv}(y_1 y_3 \cdots y_i \sigma^{-1}({\bf y}_{[i+3,n]}) y_{i+2} y_{i+1} \cdots y_4 y_2) \\ \\
  - \mathrm{inv}(y_1 y_3 \cdots y_{i+1}  \sigma^{-1}({\bf y}_{[i+3,n]}) y_{i+2} y_i  \cdots y_4 y_2) \\
  (i:\text{even})
\end{array}
\right. \\
&= \left \{
\begin{array}{cc}
n-i \\
(i:\text{odd}, y_iy_{i+1}=01 \quad \text{or} \quad i:\text{even}, y_iy_{i+1}=10), \\  \\
-(n-i) \\
 (i:\text{odd}, y_iy_{i+1}=10 \quad \text{or} \quad i:\text{even}, y_iy_{i+1}=01)
\end{array}
\right. \\
& \equiv \left \{
\begin{array}{cc}
n-i &(\tilde{y}_i \tilde{y}_{i+1}=00), \\
m-(n-i) & (\tilde{y}_i \tilde{y}_{i+1}=11).
\end{array}
\right. \\
  \end{align*}
We show that the equality holds in the case of $\tilde{y}_i \tilde{y}_{i+1}=00$ and 
$\tilde{y}_i \tilde{y}_{i+1}=11$. 
Since $i \in [n-1]$, we have
$1 \leq n-i \leq n-1 < m$. Then,  
$0 < m-(n-i) \leq m-1 < m$.
Therefore,
\begin{equation*}
r=
\begin{cases}
n-i &(\tilde{y}_i \tilde{y}_{i+1}=00), \\
m-(n-i) & (\tilde{y}_i \tilde{y}_{i+1}=11)
\end{cases}
\end{equation*}
and $r \neq 0$.
\end{proof}

\begin{proof}[proof 4 of Theorem III.3]
Lemma \ref{remainder_BAR} implies $r \neq 0$. Therefore, $\mathrm{Dec}({\bf y}) \neq {\bf y}$.
Let us focus on the step 14 of Algorithm 2.
Lemma \ref{remainder_BAR} implies $r=n-i$ or $r=m-(n-i)$. Whichever $r=n-i$ or $r=m-(n-i)$, 
\begin{align*}
\min \{r,m-r\}&=\min \{n-i, m-(n-i)\} \\
&=n-i
\end{align*}
holds, since $m \geq 2n$. Then, 
\begin{align*}
p&=n- \min \{r,m-r\} \\
  &=n-(n-i) \\
  &=i
\end{align*}
holds.  Thus,
\begin{align*}
\BAR_{p}({\bf y}) &= \BAR_{i}({\bf y}) \\
&={\bf x} \\
& \in \azinv .
\end{align*} 
Therefore, $\mathrm{Dec}({\bf y})=\BAR_{p}({\bf y})={\bf x}$.
\end{proof}

\section{Analysis and unification of the algorithms}\label{comparison}
In this section, we provide a unified view of our proposed algorithms.
After that, we prove that the algorithms can be computed 
in linear time in the code-length.  

Flowchart 1 is the unified representation of 
the parts of Algorithm 1 and Algorithm 2, 
that is, the steps 21st - 34th of Algorithm 1 and the steps 21st - 42th of Algorithm 2.
\begin{figure}[htbp]
\includegraphics[width=18cm,bb=0 0 1890 530]{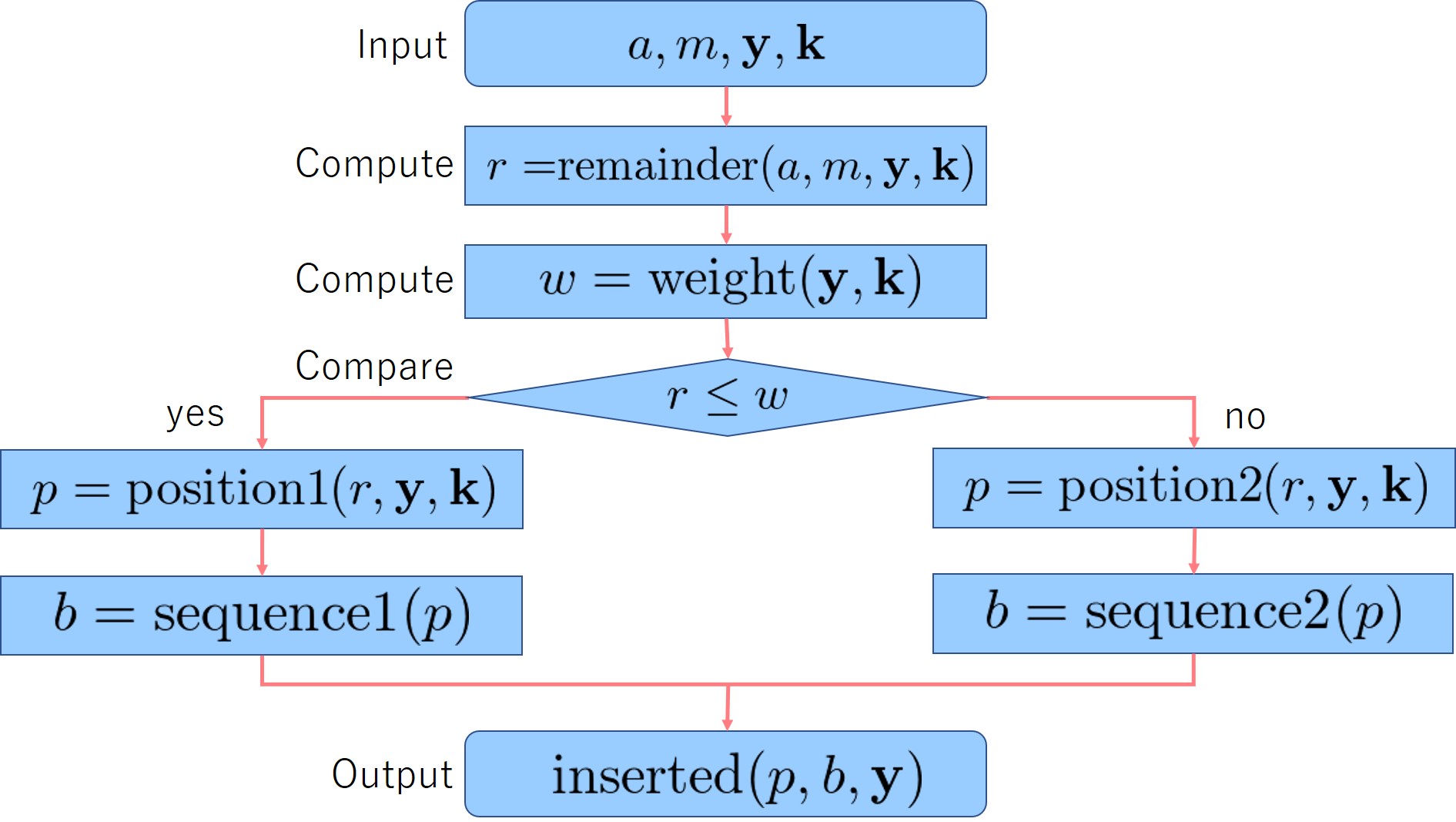}
\caption{Flowchart 1 for deletion error-correction}
\label{flow1}
\end{figure}

The following table summarizes the variables and the functions in Flowchart 1.

\begin{table}[htb]
\begin{center}
\begin{tabular}{|c|c|c|} \hline
& Monotone codes & Azinv codes \\ \hline
the condition of $a$ & $\in \mathbb{Z}_{\geq 0}$ & $\in \mathbb{Z}_{\geq 0}$ \\
the condition of $m$ & $\geq k_n +1$ & $\geq n$ \\
the condition of ${\bf y}$   &  $\in \binset^{n-1}$    & $\in \binset^{n-2}$   \\
the condition of ${\bf k}$&$(k_1,k_2,\cdots, k_n)$&$(1,2,\cdots, n)$\\
$\mathrm{remainder}(a,m,{\bf y},{\bf k})$  & $(a-\rho_{\bf k}({\bf y})) \% m$   & $(a-\tau({\bf y})) \% m$  \\
$\mathrm{weight}({\bf y},{\bf k})$ & $\mathrm{wt}_{\bf k}({\bf y})$ & $\mathrm{wt}({\bf \tilde{y}})$  \\ 
$\mathrm{position1}(r,{\bf y},{\bf k})$& $\max J_{1,1}$ & $\max J_{1,2}$ \\
$\mathrm{position2}(r,{\bf y},{\bf k})$& $\min J_{2,1}$ &  $\min J_{2,2}$ \\
$\mathrm{sequence1}(p)$ & $0$ &  
$\begin{cases}
10(p:\text{even})\\
01(p:\text{odd})
\end{cases}$ \\ 
$\mathrm{sequence2}(p)$ & $1$         &
$\begin{cases}
01(p:\text{even})\\
10(p:\text{odd})
\end{cases}$              \\ 
$\mathrm{inserted}(p,b,{\bf y})$ & $\ins_{p,{\bf b}}({\bf y})$ & $\ins_{p,{\bf b}}({\bf y})$ \\ \hline
  \end{tabular}
  \caption{Variables and Functions in Flowchart 1}
  \end{center}
\end{table}

Here, $(k_1,k_2,\cdots, k_n)$ is  a positive monotonic increasing integer sequence and 
$J_{1,1}$, $J_{1,2}$, $J_{2,1}$ and $J_{2,2}$ are defined as follows.
\begin{align*}
&J_{1,1} \coloneqq 
\{j \in [n] \mid  R^{(1)}_{{\bf k}}(j,{\bf y})=r \}, \\
&J_{1,2} \coloneqq
\{j \in [n-1] \mid L^{(1)}(j,{\bf \tilde{y}})=r \}, \\
&J_{2,1} \coloneqq
\{j \in [n] \mid L^{(0)}_{{\bf k}}(j,{\bf y})=r-w-k_1 \}, \\
&J_{2,2} \coloneqq
\{j \in [n-1] \mid R^{(0)}(j,{\bf \tilde{y}})=r-w-1 \}.
\end{align*}
For integers $a$ and $b$,  $a\% b$ denotes the remainder of $a$ divided by $b$.

Flowchart 2 is the unified representation of 
the parts of Algorithm 1 and Algorithm 2, 
that is, the steps 10th - 20th of Algorithm 1 and Algorithm 2.

\begin{figure}[htbp]
\includegraphics[width=18cm,bb=0 0 2037 403]{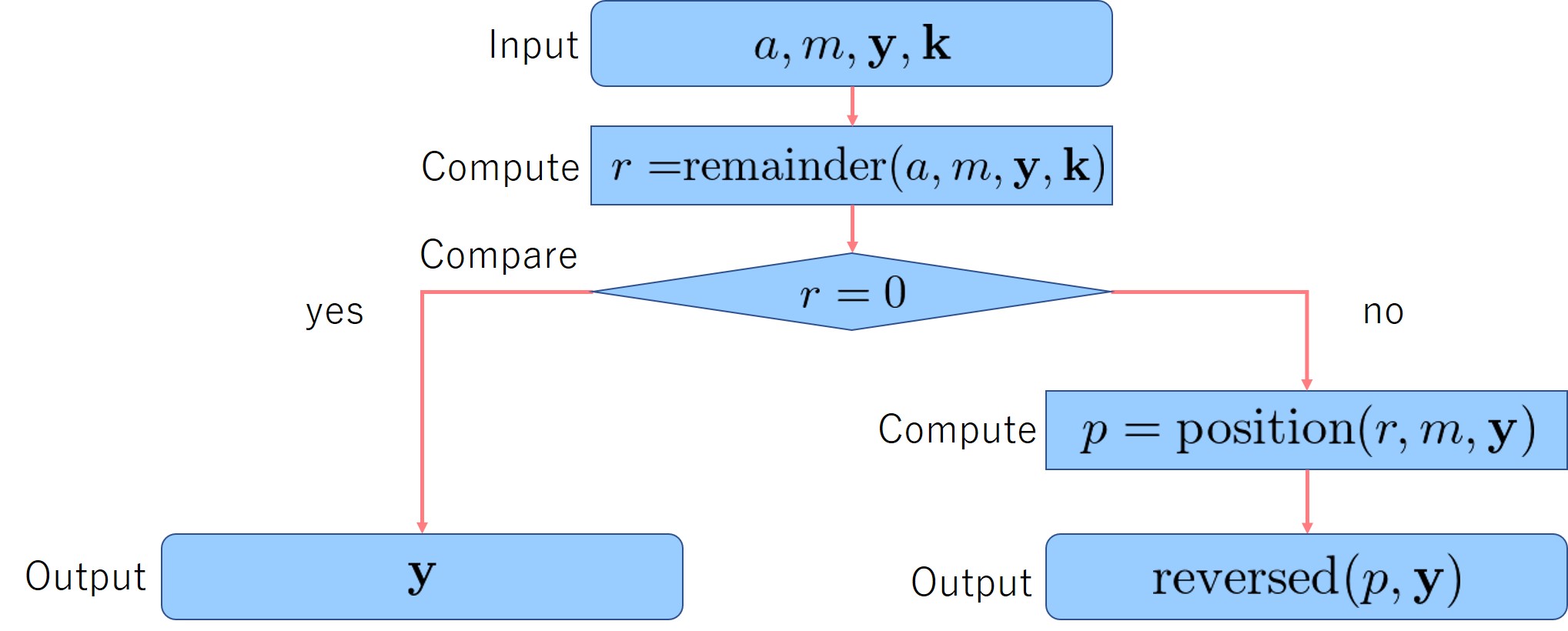}
\caption{Flowchart 2 for reversal error-correction}
\label{flow2}
\end{figure}

The following table summarizes the variables and the functions in Flowchart 2, where  
$(k_1,k_2,\cdots, k_n)$ is  a positive monotonic increasing integer sequence and
$J$ is defined as follows.
\begin{align*}
&J \coloneqq
\{ j \in[n] \mid k_j =\min\{r,m-r\} \}.
\end{align*}

\begin{table}[htb]
\begin{center}
\begin{tabular}{|c|c|c|} \hline
& Monotone codes & Azinv codes \\ \hline
the condition of $a$ & $\in \mathbb{Z}_{\geq 0}$ & $\in \mathbb{Z}_{\geq 0}$ \\
the condition of $m$ & $\geq 2k_n$ & $\geq 2(n-1)$ \\
the condition of ${\bf y}$   &  $\in \binset^{n}$   & $\in \binset^{n}$   \\
the condition of ${\bf k}$&$(k_1,k_2,\cdots, k_n)$&$(1,2,\cdots, n)$\\
$\mathrm{remainder}(a,m,{\bf y},{\bf k})$  & $(a-\rho_{\bf k}({\bf y})) \% m$   & $(a-\tau({\bf y})) \% m$  \\ 
$\mathrm{position}(r,{\bf y},{\bf k})$& $\min J$ & $n-\min \{r,m-r\}$ \\
$\mathrm{reversed}(p,{\bf y})$ & $\rev_{p}({\bf y})$ & $\BAR_{p}({\bf y})$ \\ \hline
  \end{tabular}
  \caption{Variables and Functions in Flowchart 2}
  \end{center}
\end{table}

We can compute $\mathrm{inv}({\bf y})$ in linear time in the length of ${\bf y}$
 by using the following Theorem \ref{inversion_formula}.
\begin{thm}\label{inversion_formula}
Set ${\bf s} \coloneqq (n,n-1,\cdots ,2,1)$. For ${\bf y} \in \binset^{n}$,  
\begin{align*}
\mathrm{inv}({\bf y}) = \rho_{\bf s}({\bf y}) - {{\mathrm{wt}({\bf y})+1} \choose 2}.
\end{align*}
\end{thm}
\begin{proof}
Since ${\bf y} \in \binset^n$, 
\begin{align*}
\mathrm{inv}({\bf y})
&= \# \{(i,j) \mid 1\leq i < j \leq n, y_i > y_j \}  \\
&= \# \{(i,j) \mid 1\leq i < j \leq n, (y_i, y_j)=(1,0) \}  \\
\end{align*}
holds. Set $I \coloneqq \{ i \in [n] \mid y_i=1 \}$. Then,  
we have
\begin{align*}
&\# \{(i,j) \mid 1\leq i < j \leq n, (y_i, y_j)=(1,0) \} \\
&= \sum_{i \in I} \# \{j \in [n] \mid i < j, y_j=0 \}  \\
&=\sum_{i \in I} ((n-i) -  \# \{j \in [n] \mid i < j, y_j=1 \})  \\
&=\sum_{i \in I} ((n-i+1) -  (1+\# \{j \in [n] \mid i < j, y_j=1 \}))  \\
&=\sum_{i \in I} (n-i+1) -  \sum_{i \in I} \# \{j \in [n] \mid i \leq j, y_j=1 \}  \\
&=\sum_{i \in I} y_i(n-i+1) - \frac{\# I(\# I +1)}{2}  \\
&=\rho_{\bf s}({\bf y}) - {{\mathrm{wt}({\bf y})+1} \choose 2}.
\end{align*}
\end{proof}

\begin{thm}\label{linear_function}
For each function in Flowcharts 1 or 2, 
its computational cost is $O(|{\bf y}|)$, where $|{\bf y}|$ is the length of ${\bf y}$.
\end{thm}
\begin{proof}
Since the definitions of $\mathrm{sequence1}(p)$, $\mathrm{sequence2}(p)$, 
and $\mathrm{reversed}(p,{\bf y})$,
they can be computed in constant time.

Since we only need to use "For loop" once, 
$\mathrm{position1}(r,{\bf y},{\bf k})$, $\mathrm{position2}(r,{\bf y},{\bf k})$, 
$\mathrm{position}(r,{\bf y},{\bf k})$, $\mathrm{inserted}(p,b,{\bf y})$, and ${\bf \tilde{y}}$ 
can be computed in linear time in the length of ${\bf y}$.

Since the inner product, $\mathrm{inv}({\bf y})$, and ${\bf \tilde{y}}$ 
can be computed in linear time, 
$\mathrm{remainder}(a,m,{\bf y},{\bf k})$ and $\mathrm{weight}({\bf y},{\bf k})$ 
can be computed in linear time in the length of ${\bf y}$.
\end{proof}
The following is a corollary of Theorem \ref{linear_function}.
\begin{cor}
Algorithm 1 and Algorithm 2 are linear time algorithms in the code-length.
\end{cor}

\section{conclusion}
In this paper, we provided the single deletion/reversal error-correcting algorithm for monotone codes 
and the single BAD/BAR error-correcting algorithm for azinv codes.
Constructions of these codes are different.
However, algorithms for these codes and the proofs of Theorem II.4 and Theorem III.3 correspond to each other.

In Section \ref{comparison}, we provided the unification of these decoding algorithms 
for monotone codes and azinv codes.
The respective deletion error-correcting algorithms for monotone codes and 
azinv codes are represented by the same flowchart, and the respective reversal error-correcting algorithms 
for monotone codes and azinv codes are represented by the same flowchart.
We also showed that these algorithms are linear-time algorithms.
 
As a future work, we will consider decoding algorithms for single insertion errors 
for monotone codes and azinv codes. 
Furthermore, we will concider decoding algorithms for other deletion/reversal errors.
Monotone codes are defined by $\rho_{\bf k}({\bf x})$ and 
azinv codes are defined by $\tau(\bf x)$.
By replacing one of these functions with the other, 
we will create new codes that are capable of the other deletion error-correcting 
and the other reversal error-correcting. 
These error-correcting algorithms are expected to have the same flowcharts 
as the ones for monotone codes and azinv codes.

Moreover, since monotone codes can freely take a positive monotonic increasing integer sequence ${\bf k}$, 
it is expected to be able to add the other property to monotone codes in addition to 
the single deletion/reversal error-correctable property.
For example, it is known to be able to add properties 
of being two-deletion error-correctable\cite{helberg2002multiple} 
and easy to encode \cite{hagiwara2016ordered}.
Monotone codes are generalized by introducing parameter ${\bf k}$ into Levenshtein codes.
In the same way, the generalization with parameter ${\bf k}$ in azinv codes can be considered.
The function $\mathrm{inv}({\bf x})$ used to define azinv codes has the property of \ref{inversion_formula}.
We can generalize azinv codes by taking a positive monotonic decreasing integer sequence as ${\bf s}$ in \ref{inversion_formula}.
The generalized azinv codes are expected to be able to add properties 
of being two-BAD error-correctable and easy to encode.
In addition to these properties, there are some other similar properties in Levenshtein codes and azinv codes, such as optimality and convergence\cite{hagiwara2017perfect}.
We would like to discuss these topics in a future work for further development of our research.

\section{acknowledgement}
This paper is partially supported by KAKENHI 18H01435.

\bibliographystyle{unsrt}
\bibliography{reference}

\end{document}